\documentclass[aps,pre,showkeys,onecolum,tightenlines,notitlepage]{revtex4-1}
\usepackage{amsmath,amsfonts,amsthm}
\usepackage{paralist}
\usepackage{graphics}
\usepackage{epsfig}
\usepackage{hyperref}

\newtheorem{theorem}{Theorem}
\newtheorem{corollary}{Corollary}

\begin{document}
\title{Global asymptotic properties for a Leslie-Gower food chain model}
\thanks{First published in Mathematical Biosciences and Engineering \textbf{6}:585--590 (2009).}
\author{Andrei Korobeinikov}
\email{andrei.korobeinikov@ul.ie}
\affiliation{MACSI, Department of Mathematics and Statistics, University of
  Limerick, Limerick, Ireland}

\author{William T. Lee}
\email{william.lee@ul.ie}
\affiliation{MACSI, Department of Mathematics and Statistics, University of
  Limerick, Limerick, Ireland}

\begin{abstract}
We study global asymptotic properties of a continuous time
Leslie-Gower food chain model.  We construct a Lyapunov function which
enables us to establish global asymptotic stability of the unique
coexisting equilibrium state.
\end{abstract}




\maketitle




\bigskip

In his papers~\cite{Les_48,Les._58}, P.H. Leslie introduced a
predator-prey model where both interacting species are assumed to grow
according to the logistic law.  That is both species grow with a
rate that is initially (for small population) proportional to the
population and is limited by a carrying capacity.  The novel feature
of this model is that, while the carrying capacity for the prey is a
positive constant, the carrying capacity of the predator's
environment is proportional to the prey population.  This idea leads
to a model that is quite different from the Lotka-Volterra
predator-prey model.  Leslie's model stresses the fact that
there are upper limits to the rates of increase of both prey, $H$, and
predator, $P$, which are not recognised in the Lotka-Volterra
model. These upper limits can be approached under favourable
conditions: for the predator, when the number of prey per predator is
large; for the prey, when the number of predators (and perhaps the
number of prey also) is small.  Furthermore, the Leslie-Gower model
does not posses the ``screw symmetry'' that is inherent in the
Lotka-Volterra model.
This model was initially studied by Leslie and Gower~\cite{Leslie and Gower},
and then by Pielow~\cite{Pielow}.  In the case of continuous
time, these considerations lead to the differential
equations~\cite[p. 91]{Pielow}
\begin{equation}
\frac{dH}{dt}=(r-bH-aP)H,\qquad\frac{dP}{dt}=\left(q-c\frac{P}{H}\right)P.
\label{2}
\end{equation}
Here $H(t)$ and $P(t)$ are the prey and the predator populations
respectively; $r$ and $q$ are the growth rates of the prey and the
predator respectively; $a$ is the attack rate; $1/rb$ is the carrying
capacity of the prey environment, and $1/qc$ is the efficiency of
consumption for the predator (that is $H(t)/qc$ is the predator
population that the prey population of the size $H$ can support).  All
the constants in the system~(\ref{2}) are positive.  This model always
has the unique coexisting fixed point $Q^{*}=(H^{*},P^{*})$, where
\begin{equation}
H^{*}=\frac{rc}{aq+bc},\qquad P^{*}=\frac{rq}{aq+bc},
\label{3}
\end{equation}
which was proved to be globally asymptotically stable~\cite{Korobeinikov_2001}.

The Leslie-Gower model can be immediately extended to the case of a
food chain. The food chain composed of $n+1$ levels where the $i$th level
depends (predates) upon only the $i-1$ level can be represented by the transfer
diagram
\[
H\longrightarrow P_{1}\longrightarrow P_{2}\longrightarrow \ldots
\longrightarrow P_{n}.
\]
This food chain can be described by the following system of differential equations:
\begin{eqnarray}
    \frac{dH}{dt} & = & (r-bH-aP_{1})H, \nonumber \\ \frac{dP_{1}}{dt}
    & = & \left(q_{1}-c_{1}\frac{P_{1}}{H}-s_{1}P_{2}\right)P_{1},
    \nonumber \\ & \vdots & \nonumber \\ \frac{dP_{i}}{dt} & = &
    \left(q_{i}-c_{i}\frac{P_{i}}{P_{i-1}}-s_{i}P_{i+1}\right)P_{i},
\label{chain-model}\\
               & \vdots & \nonumber \\ \frac{dP_{n}}{dt} & = &
\left(q_{n}-c_{n}\frac{P_{n}}{P_{n-1}}\right)P_{n}. \nonumber
\end{eqnarray}
Here the parameters $q_i$, $c_i$ and $s_i$
are defined by analogy to the single predator case, namely $q_i$ is the
reproduction rate of the $i$th predator, $c_i$ is defined so
that $1/q_{i}c_{i}$ is the efficiency of consumption for the $i$th predator,
and $s_i$ is the attack rate by the $i$th predator. In order for the
equations to be biologically meaningful these parameters must all be
positive quantities.

The global properties of this model are given by the following Theorem:

\begin{theorem}
The Leslie-Gower chain model (\ref{chain-model}) always has
a positive (coexisting) equilibrium state $Q^{*}_{n}=(H^{*},P^{*}_{1},\ldots,P^{*}_{n})$;
this equilibrium state is unique and globally asymptotically stable.
\end{theorem}

\begin{proof}
\textbf{(1). Existence of the positive equilibrium state.}  We prove
this by induction.  The positive equilibrium state always exists for
$n=1$: the coordinates of the equilibrium state are given by
equalities~(\ref{3}). We assume that the statement of Theorem holds
when $n=m$ and prove that it holds for $n=m+1$ as well.

Starting from an $m$ level chain in which, by assumption, all equilibrium
populations $H^{(m)*}$, $P_i^{(m)*}$,
$i=1,\ldots,m$ are positive, we convert the
system to an $m+1$ level chain by introducing a population of
top level ($m+1$) predators and allow the system to equilibrate.
It can readily be seen that the positive region $\mathbb{R}^{m+1}_{+}$ is
an invariant set of this system. That prevents the sign
of $P_i$ for all $i=1,\ldots,m$ changing. Thus, from the
assumption of a positive equilibrium in the $m$ level system, it
follows that in the $m+1$ level system all $P_i^{(m+1)*}$ for  $i=1,\ldots,m$
are positive or zero.
Consider the final differential equation
\begin{equation}
\dfrac{\mathrm{d}P_{m+1}}{\mathrm{d}t}=(a-\dfrac{bP_{m+1}}{P_m})P_{m+1}.
\end{equation}
In order for equilibrium population at $P^{*}_{m+1}=0$ (rather than $P^{*}_{m+1}>0$) to
be attained we must have
\begin{equation}
\lim_{P^{*}_{m+1}\rightarrow 0} a-\dfrac{bP^{*}_{m+1}}{P^{*}_m}\leq 0.
\end{equation}
This requires $\lim_{P^{*}_{m+1}\rightarrow 0} P^{*}_m\leq 0$, whereas by
assumption that the $m$-level system has a positive equilibrium state
the converse holds. Thus the existence of an $m$ level positive
equilibrium state implies the existence of an $m+1$ level positive
equilibrium state. This completes this section of the proof.

\begin{figure}
\begin{center}
\includegraphics[width=0.8\textwidth]{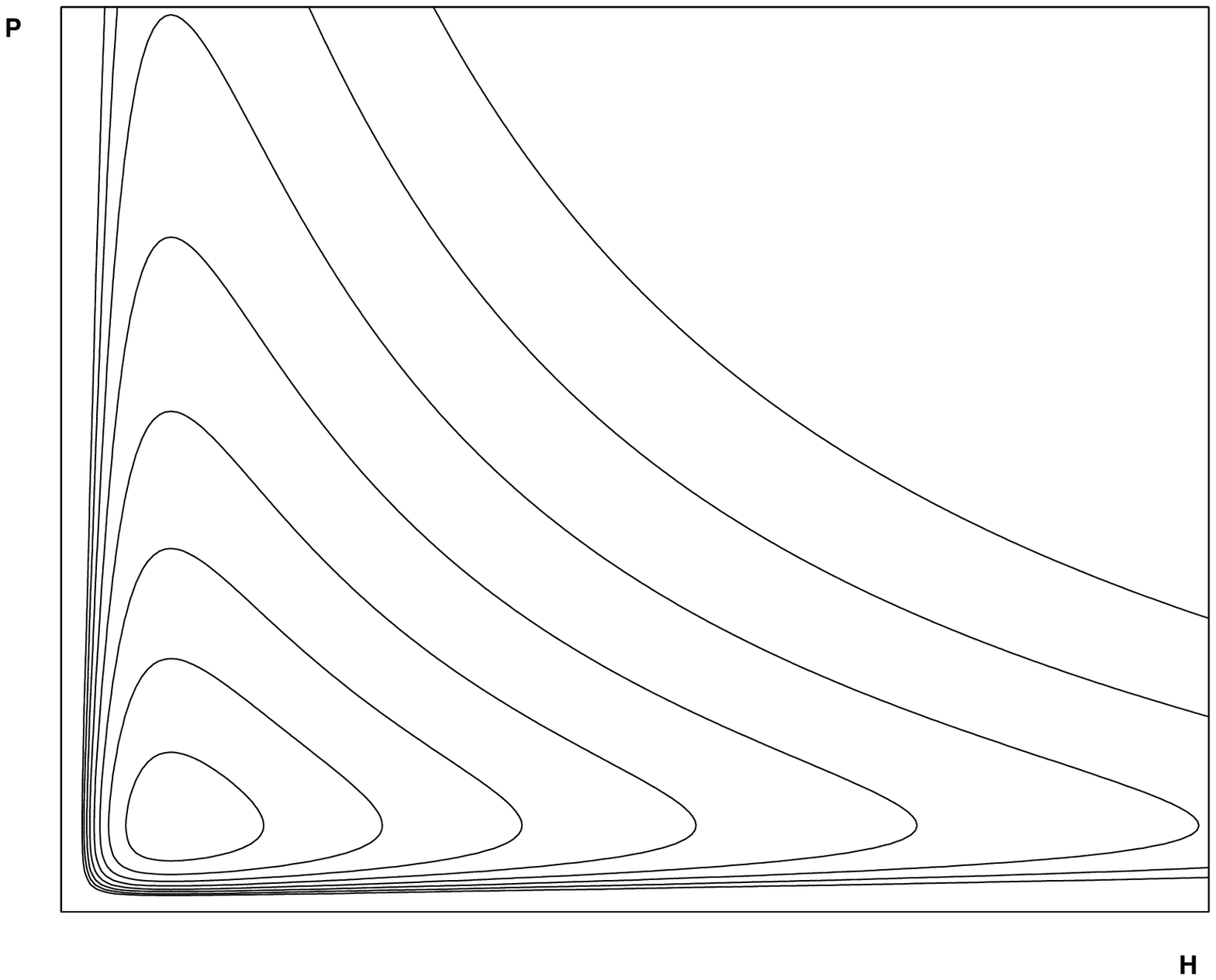}
\end{center}
\caption{Level curves of the Lyapunov function
\mbox{$\ln\frac{H}{H^{*}}+\frac{H^{*}}{H}+B\left(\ln\frac{P}{P^{*}}+\frac{P^{*}}{P}\right)$.}}
\end{figure}

\textbf{(2). Global asymptotic stability of the positive equilibrium state.}
A Lyapunov function
\[
V(H,P_{1},\ldots,P_{n})=
\left(\ln\frac{H}{H^{*}}+\frac{H^{*}}{H}\right)
+\sum_{i=1}^{n}B_{i}\left(\ln\frac{P_{i}}{P_{i}^{*}}+\frac{P_{i}^{*}}{P_{i}}\right),
\]
where $B_{i}c_{i}=B_{i-1}s_{i-1}P_{i-1}^{*}$ and $B_{1}c_{1}=aH^{*}$,
is defined and continuous for all $H,P_{1},\ldots,P_{n}>0$.
The function $V(H,P_{1},\ldots,P_{n})$ satisfies
\[
\frac{\partial V}{\partial H}=\frac{1}{H}\left(1-\frac{H^{*}}{H}\right),
\qquad
\frac{\partial V}{\partial P_{i}}=\frac{B_{i}}{P_{i}}\left(1-\frac{P_{i}^{*}}{P_{i}}\right),
\]
and hence the fixed point $Q^{*}_{n}$ is the only extremum of this function.
It is easy to see that the point $Q^{*}_{n}$ is the global minimum
of $V(H,P_{1},\ldots,P_{n})$ in $\mathbb{R}^{n+1}_{+}$.
(Fig. 1 shows the level curves of this function for $n=$1.)

The function $V(H,P_{1},\ldots,P_{n})$ satisfies
\begin{eqnarray*}
\frac{dV}{dt} & = & r-bH-aP_{1}-r\frac{H^{*}}{H}+bH^{*}+a\frac{H^{*}P_{1}}{H}\\
&  & +B_{1}\left(q_{1}-c_{1}\frac{P_{1}}{H}-s_{1}P_{2}-q_{1}\frac{P_{1}^{*}}{P_{1}}
          +c_{1}\frac{P_{1}^{*}}{H}+s_{1}\frac{P_{1}^{*}P_{2}}{P_{1}}\right)\\
&  & +\sum_{i=2}^{n-1}B_{i}\left(q_{i}-c_{i}\frac{P_{i}}{P_{i-1}}-s_{i}P_{i+1}
          -q_{i}\frac{P_{i}^{*}}{P_{i}}+c_{i}\frac{P_{i}^{*}}{P_{i-1}}
          +s_{i}\frac{P_{i}^{*}P_{i+1}}{P_{i}}\right)\\
&  & +B_{n}\left(q_{n}-c_{n}\frac{P_{n}}{P_{n-1}}-q_{n}\frac{P_{n}^{*}}{P_{n}}
          +c_{n}\frac{P_{n}^{*}}{P_{n-1}}\right)\\
& = & r+\sum_{i=1}^{n}B_{i}q_{i}+ bH^{*}-bH-r\frac{H^{*}}{H}+B_{1}c_{1}\frac{P_{1}^{*}}{H}\\
&   & -aP_{1}-\sum_{i=2}^{n}B_{i-1}s_{i-1}P_{i}
      -\sum_{i=1}^{n}B_{i}q_{i}\frac{P_{i}^{*}}{P_{i}}+\sum_{i=2}^{n}B_{i}c_{i}\frac{P_{i}^{*}}{P_{i-1}}\\
&   & +a\frac{H^{*}P_{1}}{H}-B_{1}c_{1}\frac{P_{1}}{H}
          +\sum_{i=1}^{n-1}B_{i}s_{i}\frac{P_{i}^{*}P_{i+1}}{P_{i}}
          -\sum_{i=2}^{n}B_{i}c_{i}\frac{P_{i}}{P_{i-1}}.
\end{eqnarray*}
By the definition of $B_{i}$, the equalities
\[  a\frac{H^{*}P_{1}}{H}-B_{1}c_{1}\frac{P_{1}}{H}=0, \qquad
    \sum_{i=1}^{n-1}B_{i}s_{i}\frac{P_{i}^{*}P_{i+1}}{P_{i}}
        -\sum_{i=2}^{n}B_{i}c_{i}\frac{P_{i}}{P_{i-1}}=0\]
hold. Furthermore, recollecting that
\[
r=bH^{*}+aP_{1}^{*}, \quad
q_{i}=c_{i}\frac{P_{i}^{*}}{P_{i-1}^{*}}+s_{i}P_{i+1}^{*}, \quad
q_{n}=c_{n}\frac{P_{n}^{*}}{P_{n-1}^{*}}
\]
hold at $Q^{*}_{n}$, we obtain
\[
\sum_{i=1}^{n}B_{i}q_{i}
 = B_{1}c_{1}\frac{P_{1}^{*}}{H^{*}} + \sum_{i=2}^{n}B_{i}c_{i}\frac{P_{i}^{*}}{P_{i-1}^{*}} + \sum_{i=1}^{n}B_{i}s_{i}P_{i+1}^{*}
 = aP_{1}^{*}+ 2\sum_{i=2}^{n}B_{i-1}s_{i-1}{P_{i}^{*}},
\]
\[ r\frac{H^{*}}{H}-B_{1}c_{1}\frac{P_{1}^{*}}{H}=bH^{*}\frac{H^{*}}{H}, \]
and
\begin{eqnarray*}
\sum_{i=1}^{n}B_{i}q_{i}\frac{P_{i}^{*}}{P_{i}}-\sum_{i=2}^{n}B_{i}c_{i}\frac{P_{i}^{*}}{P_{i-1}} & = &
B_{n}q_{n}\frac{P_{n}^{*}}{P_{n}}+\sum_{i=1}^{n-1}B_{i}\left(q_{i}-s_{i}P_{i+1}^{*}\right)\frac{P_{i}^{*}}{P_{i}}\\
& = & aP_{1}^{*}\frac{P_{1}^{*}}{P_{1}}+\sum_{i=2}^{n}B_{i}c_{i}\frac{P_{i}^{*}}{P_{i-1}^{*}}\frac{P_{i}^{*}}{P_{i}}\\
& = & aP_{1}^{*}\frac{P_{1}^{*}}{P_{1}}+\sum_{i=2}^{n}B_{i-1}s_{i-1}P_{i}^{*}\frac{P_{i}^{*}}{P_{i}}.
\end{eqnarray*}
Therefore,
\begin{eqnarray*}
\frac{dV}{dt} & = & bH^{*}\left(2-\frac{H^{*}}{H}-\frac{H}{H^{*}}\right)\\
  &&  +aP_{1}^{*}\left(2-\frac{P_{1}^{*}}{P_{1}}-\frac{P_{1}}{P_{1}^{*}}\right)
          +\sum_{i=2}^{n}B_{i-1}s_{i-1}P_{i}^{*}
                \left(2-\frac{P_{i}^{*}}{P_{i}}-\frac{P_{i}}{P_{i}^{*}}\right)\\
& = & -bH\left(1-\frac{H^{*}}{H}\right)^2+aP_{1}\left(1-\frac{P_{1}^{*}}{P_{1}}\right)^2
          +\sum_{i=2}^{n}B_{i-1}s_{i-1}P_{i}\left(1-\frac{P_{i}^{*}}{P_{i}}\right)^2.
\end{eqnarray*}
That is, for this model $\frac{dV}{dt}<0$ strictly holds for all
$H,P_{1},\ldots,P_{n}>0$, except the fixed point $Q^{*}_{n}$ where $\frac{dV}{dt}=0$.
Therefore, by the Lyapunov asymptotic stability theorem~\cite{Lyapunov},
the fixed point $Q^{*}$ is globally asymptotically stable.

\textbf{(3). Uniqueness of the positive equilibrium state.}
At any equilibrium state, $\frac{dV}{dt}=0$ must hold.
For this model, however, the fixed point $Q^{*}_{n}$ is the only point
in $\mathbb{R}_{+}^{n+1}$ where $\frac{dV}{dt}=0$ holds.

This completes the proof.
\end{proof}

Apart from the positive equilibrium state $Q^{*}_{n}$ where all $n+1$
species coexist, this system also has $n$ equilibrium states
$Q^{*}_{k}$ (where $k=0,1,\ldots,n-1$), which corresponds to the
reduced $m$-species food chains
\[
H\longrightarrow P_{1}\longrightarrow P_{2}\longrightarrow \ldots
\longrightarrow P_{k}.
\]
Biologically these correspond to the case in which some external
intervention has reduced the population of the $i$th species to zero (we
have proved above that this system is uniformly persistent,
and hence that can never occur for this model via the natural evolution
of the system) leading to the extinction of the $i+1$th level species that feeds
on species $i$, and then to all higher levels of the food chain.
For each of these equilibrium
states, $H,P_{1},\ldots,P_{k}>0$ while $P_{k+1}=\ldots=P_{n}=0$.
Thus, $Q^{*}_{0}$ corresponds to the predator-free case and has the
coordinates $H_{0}=r/b, P_{1}=\ldots=P_{n}=0$; $Q^{*}_{1}$ coincides
with the equilibrium state (\ref{3}) of the two-species
model~(\ref{2}).  The following Corollary immediately follows from the
Theorem:

\begin{corollary}
Apart from the positive equilibrium state $Q^{*}_{n}$, the system has
$n$ non-negative equilibrium states $Q^{*}_{k}$ (where
$k=0,1,\ldots,n-1$).  Each of these equilibrium states is unstable in
$\mathbb{R}^{n}_{\geq 0}$, but globally asymptotically stable in the
$k$-dimensional invariant subspace $\mathbb{R}^{k}_{+}=\{
H,P_{1},\ldots,P_{k}>0;P_{k+1}=\ldots=P_{n}=0 \} $.
\end{corollary}

In conclusion, we have to note that, apart from the mentioned $n+1$
equilibrium states that are located in the nonnegative region
$\mathbb{R}^{n}_{\geq 0}$, the system has other $n-1$ points with the
coordinates that satisfy the equalities
\[
r=bH+aP_{1}, \quad
q_{i}=c_{i}\frac{P_{i}}{P_{i-1}}+s_{i}P_{i+1}, \quad
q_{n}=c_{n}\frac{P_{n}}{P_{n-1}}.
\]
Indeed, it is readily seen that this system of algebraic equations
is equivalent to a polynomial of the degree $n$ and that this system
has no complex solutions.
However, the existence of these equilibria do not contradict the
Theorem since these points are located outside of the non-negative
region $\mathbb{R}^{n}_{\geq 0}$, which is the phase space of the
system.  The origin is an unstable equilibrium state of the system as
well.

These results demonstrate both the practicality and the usefulness of
performing a stability analysis on non-trivial ecosystem models. We
have shown, using a Leslie-Gower food chain model as an
example, that it is possible to enumerate and characterize the
stability properties of all the equilibrium states of the model.

\begin{acknowledgments}
We acknowledge support of the Mathematics Applications Consortium for
Science and Industry (\url{www.macsi.ul.ie}) funded by the Science
Foundation Ireland Mathematics Initiative Grant 06/MI/005.
\end{acknowledgments}

\medskip

\medskip
\end{document}